\documentclass[letterpaper, 10 pt, conference]{ieeeconf}
\IEEEoverridecommandlockouts
\usepackage{cite}

\usepackage{color}

\usepackage[pdftex]{graphicx}
\graphicspath{{Figs/}{../pdf/}{../jpeg/}}
\DeclareGraphicsExtensions{.pdf,.jpeg,.png}

\usepackage{amsmath}

\usepackage{amssymb}

\usepackage{mathrsfs}

\usepackage{algorithm}

\usepackage{array}

\usepackage{fixltx2e}

\usepackage{stfloats}

\usepackage{url}

\usepackage{booktabs}

\usepackage{enumerate}

\usepackage{setspace}

\usepackage{amsmath}
\newtheorem{theorem}{Theorem}[section]
\newtheorem{lemma}[theorem]{Lemma}

\definecolor{ORANGE}{RGB}{234,122,13}
\definecolor{BLUE}{RGB}{0,112,192}
\definecolor{DRED}{RGB}{192,0,0}
\definecolor{GREEN1}{RGB}{73,156,80}

\begin{document}

\title{Data-Enabled Predictive Control for Grid-Connected Power Converters}

\author{Linbin Huang, Jeremy Coulson, John Lygeros and Florian D{\"o}rfler
\thanks{L. Huang is with the College of Electrical Engineering at Zhejiang University, Hangzhou, China, and the Department of Information Technology and Electrical Engineering at ETH Z{\"u}rich, Switzerland. (Email: huanglb@zju.edu.cn)}
\thanks{J. Coulson, J. Lygeros and F. D{\"o}rfler are with the Department of Information Technology and Electrical Engineering at ETH Z{\"u}rich, Switzerland. (Emails: jcoulson@control.ee.ethz.ch, jlygeros@ethz.ch, dorfler@ethz.ch)}
\thanks{This research was supported by ETH Z{\"u}rich Funds.}}


\maketitle


\begin{abstract}
We apply a novel data-enabled predictive control (DeePC) algorithm in grid-connected power converters to perform safe and optimal control. Rather than a model, the DeePC algorithm solely needs input/output data measured from the unknown system to predict future trajectories. We show that the DeePC can eliminate undesired oscillations in a grid-connected power converter and stabilize an unstable system. However, the DeePC algorithm may suffer from poor scalability when applied in high-order systems. To this end, we present a finite-horizon output-based model predictive control (MPC) for grid-connected power converters, which uses an {\em N}-step auto-regressive-moving-average (ARMA) model for system representation. The ARMA model is identified via an {\em N}-step prediction error method (PEM) in a recursive way. We investigate the connection between the DeePC and the concatenated PEM-MPC method, and then analytically and numerically compare their  closed-loop performance. Moreover, the PEM-MPC is applied in a voltage source converter based HVDC station which is connected to a two-area power system so as to eliminate low-frequency oscillations. All of our results are illustrated with high-fidelity, nonlinear, and noisy simulations.
\end{abstract}



\section{Introduction}

The penetration of power-electronic devices in modern power systems is ever-increasing due to the development of renewable energy, microgrids, high-voltage direct-current (HVDC) transmission systems, etc. \cite{FM-FD-GH-DH-GV:18,olivares2014trends}. This tendency is posing great challenges to power system operations because the dynamics of power converters are substantially different from synchronous generators (SGs). For example, SGs have large rotational rotors which physically determine the output frequencies, while power converters consist of static semiconductor apparatus and have high controllability. 

Conventionally, the control structure of power converters is designed according to engineering experience and the corresponding control gain tuning is based on iterative trial-and-error methods. Also, lots of effort has been put into the modeling of power converters, which provides insights into the system dynamics and criteria for control gain tuning \cite{pogaku2007modeling,harnefors2007modeling,cespedes2014impedance}. However, these approaches heavily rely on rich engineering experience and lack systematicness. In addition, the control structure generally assumes a stiff power grid and may present poor robustness against variable grid conditions. For example, the most widely-used control structure, which consists of a phase-locked loop (PLL) and a current control loop, can become unstable when the power converter is connected to a weak grid with high grid impedance (or equivalently, low short-circuit ratio) \cite{wen2016analysis,suul2016impedance,huang2018adaptive}.

Even though offline design and analysis (based on a nominal model) can be conducted to determine an optimal control parameter set, optimal performance can rarely be achieved during online operation because (i) the real parameters of the power converter (e.g., capacitance and inductances of the {\em{LCL}}  filter) are hard to obtain due to different operation conditions and manufacturing inaccuracy; (ii) sometimes the underlying algorithms for the converter are designed by another manufacturer and are not obtainable, i.e., some part of the converter system is unknown; (iii) the power grid is generally an unknown system from power converter side which significantly affects the dynamic performance; and (iv) the offline design generally employs a constant power grid model (which in most cases is assumed to be an infinite bus) for the power converter, yet the real power grid is variable.
%

Normally, these problems are handled using robust or adaptive methods \cite{weiss2004h,huang2018adaptive}. However, these methods are still model-based, result in complex controllers, and suffer from scalability problems for large and uncertain (or even partially unknown) models -- especially, in grid-connected applications. Inspired by recent advances in machine learning and artificial intelligence, recent control approaches entirely circumvent such model-based solutions in favor of data-driven approaches \cite{lewis2012reinforcement,dean2017sample,boczar2018finite}.

In this paper, we use a novel {\bf{D}}ata-{\bf{e}}nabl{\bf{e}}d {\bf{P}}redictive {\bf{C}}ontrol ({\bf DeePC}) algorithm to compute optimal and safe control policies for grid-connected power converters, which uses real-time feedback to drive the unknown system along a desired (i.e., optimal and constrained) trajectory \cite{coulson2018data}. The DeePC algorithm presented in \cite{coulson2018data} relies on behavioural system approach \cite{markovsky2006exact, willems2005note, markovsky2008data, markovsky2005algorithms}. Instead of using a parametric model for system representation (e.g., state space matrices obtained from system identification), the approach in \cite{markovsky2006exact, willems2005note, markovsky2008data, markovsky2005algorithms} describes the input/output behaviour of the system through the subspace of the signal space in which trajectories live. This signal space of trajectories is spanned by the columns of a data Hankel matrix which results in a non-parametric and data-centric perspective on dynamical control systems. 

The DeePC approach presented here relies on input/output data samples from the converter-internal signals and terminal signals towards the power grid (whose model is unknown from the perspective of the converter's controller), and successfully eliminates undesired oscillations by solving the optimal regulation problem in a receding horizon manner -- with the input/output constraints incorporated and in absence of any system model. However, when applied in large-scale systems, e.g., in the case of power transmision oscillation damping \cite{bjork2019performance,huang2018damping}, the optimal regulation problem in DeePC may suffer from poor scalability due to its high dimension.

To this end, we use a finite-horizon output-based {\em model predictive control} ({\bf MPC}) for grid-connected power converters, wherein the unknown system is represented by an $N$-step {\em auto-regressive-moving-average} (ARMA) model and identified via least-square $N$-step {\em prediction error method} ({\bf PEM}). The PEM can be solved in a recursive way which enables an iterative calculation and possible online implementation. We will show that this concatenated PEM-MPC method is scalable for large-scale unknown systems, and analytically discuss how it is related to DeePC. Namely, DeePC provably outperforms the PEM-MPC method in terms of the cost in the optimal regulation problem (see Lemma \ref{Lemma: PEM-MPC and DeePC}) although this performance gap can be made smaller by appropriate regularizations of the optimal control and system identification problems.
We propose to use the PEM-MPC in a grid-connected voltage source converter (VSC) based HVDC station to attenuate the power system oscillations which are caused by the interactions among multiple synchronous generators \cite{kundur1994power}. All of our results are illustrated with high-fidelity nonlinear simulations.

The remainder of this paper is organized as follows: in Section II we provide an overview for the DeePC approach and apply it to a grid-connected power converter. Section III presents the concatenated PEM-MPC and we discuss how it is related to DeePC. In Section IV we apply the PEM-MPC in a two-area power system which contains one VSC-HVDC station. We conclude the paper in Section V.

\section{Data-Enabled Predictive Control}
\subsection{Preliminaries and Notation}
For an unknown discrete-time LTI system that has $m$ inputs and $p$ outputs, we denote by $u_{i,t} \in \mathbb{R}$ the $i\rm{th}$ input of the system at time $t \in \mathbb{Z}_{ \ge 0}$ and $y_{i,t} \in \mathbb{R}$ the $i\rm{th}$ output at time $t \in \mathbb{Z}_{ \ge 0}$, where $\mathbb{Z}_{ \ge 0}$ is the discrete-time axis. The input vector of the system is denoted by $u_t = {\rm{col}}(u_{1,t},...,u_{m,t}) \in \mathbb{R}^{m}$, and the output vector is denoted by $y_t = {\rm{col}}(y_{1,t},...,y_{p,t}) \in \mathbb{R}^{p}$, where ${\rm{col}}(a_1,...,a_i):=[a_1^{\top}\; \cdots \;a_i^{\top}]^{\top}$. Let $u = {\rm{col}}(u_1,u_2,...)$ and $y = {\rm{col}}(y_1,y_2,...)$ be the input and output trajectories, respectively, whose dimensions can be inferred from the context.

Let $L,T \in \mathbb{Z}_{ \ge 0}$ and $T \ge L$. The trajectory $u \in \mathbb{R}^{mT}$ is \textit{persistently exciting of order L} if the Hankel matrix
\begin{equation}
\mathscr{H}_L(u) := \left[ {\begin{array}{*{20}{c}}
	{{u_1}}&{{u_2}}& \cdots &{{u_{T - L + 1}}}\\
	{{u_2}}&{{u_3}}& \cdots &{{u_{T - L + 2}}}\\
	\vdots & \vdots & \ddots & \vdots \\
	{{u_L}}&{{u_{L + 1}}}& \cdots &{{u_T}}
	\end{array}} \right]		\label{eq:Hankel_L}
\end{equation}
is of full row rank, i.e., the signal $u$ is sufficiently long and sufficiently rich.

Consider the following $n$-order discrete-time LTI system (minimal representation):
\begin{equation}
\left\{ \begin{array}{l}
{x_{t + 1}} = A{x_t} + B{u_t}\\
{y_t} = C{x_t} + D{u_t}
\end{array} \right.\,,		\label{eq:ABCD}
\end{equation}
where $A \in \mathbb{R}^{n \times n}$, $B \in \mathbb{R}^{n \times m}$, $C \in \mathbb{R}^{p \times n}$, $D \in \mathbb{R}^{p \times m}$, and $x_t$ is the state of the system at time $t \in \mathbb{Z}_{ \ge 0}$.

The \textit{lag} of the system in (\ref{eq:ABCD}) is defined by the smallest integer $\ell \in \mathbb{Z}_{ \ge 0}$ so that the observability matrix
\begin{equation*}
\mathscr{O}_{\ell}(A,C) := {\rm{col}}(C,CA,...,CA^{\ell-1})
\end{equation*}
has rank $n$.

Let $T_{\rm ini},N \in \mathbb{Z}_{ \ge 0}$ such that $T \ge (m+1)(T_{\rm ini}+N+n)-1$. Consider an input trajectory $u^{\rm{d}}$ and an output trajectory $y^{\rm{d}}$ (both are of length $T \in \mathbb{Z}_{ \ge 0}$, i.e., $u^{\rm{d}} \in \mathbb{R}^{mT}$ and $y^{\rm{d}} \in \mathbb{R}^{pT}$) measured from the $n$-order unknown system (\ref{eq:ABCD}) such that $u^{\rm{d}}$ is persistently exciting of order $T_{\rm ini} + N + n$. Here we use the superscript d to indicate that these two trajectories are data sets measured from the unknown system. We use $u^{\rm{d}}$ and $y^{\rm{d}}$ to construct the Hankel matrices $\mathscr{H}_{T_{\rm ini}+N}(u^{\rm{d}})$ and $\mathscr{H}_{T_{\rm ini}+N}(y^{\rm{d}})$, which are further partitioned into two parts as
\begin{equation}
\left[ {\begin{array}{*{20}{c}}
	{{U_P}}\\
	{{U_f}}
	\end{array}} \right] := \mathscr{H}_{T_{\rm ini}+N}(u^{\rm{d}})\,,\;\;\left[ {\begin{array}{*{20}{c}}
	{{Y_P}}\\
	{{Y_f}}
	\end{array}} \right] := \mathscr{H}_{T_{\rm ini}+N}(y^{\rm{d}})\,,		\label{eq:partition_Huy}
\end{equation}
where $U_P \in \mathbb{R}^{mT_{\rm ini} \times (T-T_{\rm ini}-N+1)}$, $U_f \in \mathbb{R}^{mN \times (T-T_{\rm ini}-N+1)}$, $Y_P \in \mathbb{R}^{pT_{\rm ini} \times (T-T_{\rm ini}-N+1)}$ and $Y_f \in \mathbb{R}^{pN \times (T-T_{\rm ini}-N+1)}$.

According to the behavioral system theory \cite{willems2005note}, ${\rm{col}}(u_{\rm ini},y_{\rm ini},u,y)$ is a trajectory of (\ref{eq:ABCD}) if and only if there exist $g \in \mathbb{R}^{T-T_{\rm ini}-N+1}$ such that 
\begin{equation}
\left[ {\begin{array}{*{20}{c}}
	{{U_P}}\\
	{{Y_P}}\\
	{{U_f}}\\
	{{Y_f}}
	\end{array}} \right]g = \left[ {\begin{array}{*{20}{c}}
	{{u_{\rm ini}}}\\
	{{y_{\rm ini}}}\\
	u\\
	y
	\end{array}} \right]\,.		\label{eq:Hankel_g}
\end{equation}
The trajectory ${\rm{col}}(u_{\rm ini},y_{\rm ini})$ can be thought of as an initial condition for the trajectory ${\rm{col}}(u_{\rm ini},y_{\rm ini},u,y)$ and ${\rm{col}}(u,y)$ as a future trajectory departing from this initial condition.
If $T_{\rm ini} \ge \ell$, the future output trajectory $y$ is uniquely determined through (\ref{eq:Hankel_g}) for every given input trajectory $u$.


\subsection{Review of DeePC}

Instead of learning a parametric system representation through system identification, the DeePC attempts to learn the system's behaviour and computes optimal control inputs using past data measured from the unknown system. Moreover, input/output constraints can be conveniently incorporated to ensure safety, described as follows.

After using the input/output trajectory ${\rm{col}}(u^{\rm{d}},y^{\rm{d}})$ ($u^{\rm{d}} \in \mathbb{R}^{mT}$ and $y^{\rm{d}} \in \mathbb{R}^{pT}$) to construct the Hankel matrices in (\ref{eq:partition_Huy}), DeePC solves the following optimization problem at every sampling time to get the optimal future control inputs
\begin{equation}
\begin{array}{l}
\mathop {{\rm{min}}}\limits_{g,u \in \mathcal U, y \in \mathcal Y} \;\;{\left\| u \right\|_R^2} + {\left\| {y - r} \right\|_Q^2} + {\lambda _g}{\left\| g \right\|_2^2}\\
s.t.\;\;\left[ {\begin{array}{*{20}{c}}
	{{U_P}}\\
	{{Y_P}}\\
	{{U_f}}\\
	{{Y_f}}
	\end{array}} \right]g = \left[ {\begin{array}{*{20}{c}}
	{{u_{\rm ini}}}\\
	{{y_{\rm ini}}}\\
	u\\
	y
	\end{array}} \right]\,,		
	\tag{DeePC}
	\label{eq:DeePC}
\end{array}
\end{equation}
where $\mathcal U \subseteq \mathbb{R}^{mN}$ and $\mathcal Y \subseteq \mathbb{R}^{pN}$ are the input and output constraint sets, $R \in \mathbb{R}^{mN \times mN}$ is the control cost matrix (positive definite), $Q \in \mathbb{R}^{pN \times pN}$ is the output cost matrix (positive semidefinite), $\lambda_g \in \mathbb{R}_{ \ge 0}$ is the regularization parameter, $r \in \mathbb{R}^{pN}$ is the reference vector for the output signals, $N$ is the prediction horizon, and ${\rm{col}}(u_{\rm ini},y_{\rm ini})$ consists of the most recent input/output trajectory of (\ref{eq:ABCD}).
The norm ${\left\| a \right\|_X^2}$ of the vector $a$ denotes the quadratic form $a^TXa$, and the norm ${\left\| a \right\|_2}$ denotes $\sqrt {{a^T}a} $. 

We note that a two-norm penalty on $g$ is included in the cost function as a regularization term to avoid overfitting. In fact, in the case when stochastic disturbances affect the output measurements, a two-norm regularization on $g$ coincides with two-norm robustness with respect to the noise affecting the output measurements \cite{Coulson2019Regularized}. The DeePC involves solving the optimization problem (\ref{eq:DeePC}) in a receding horizon manner \cite{coulson2018data}, that is, after calculating the optimal control input sequence $u^\star$, we apply $(u_t,...,u_{t+s}) = (u_0^{\star},...,u_s^{\star})$ to the system for some $s \le N-1$ time steps, update ${\rm{col}}(u_{\rm ini},y_{\rm ini})$ to the most recent input/output measurements and then set $t$ to $t+s+1$ for the DeePC algorithm.

%
%

\subsection{Application to a Grid-Connected Power Converter}

\begin{figure}[!t]
	\centering
	\includegraphics[width=2.9in]{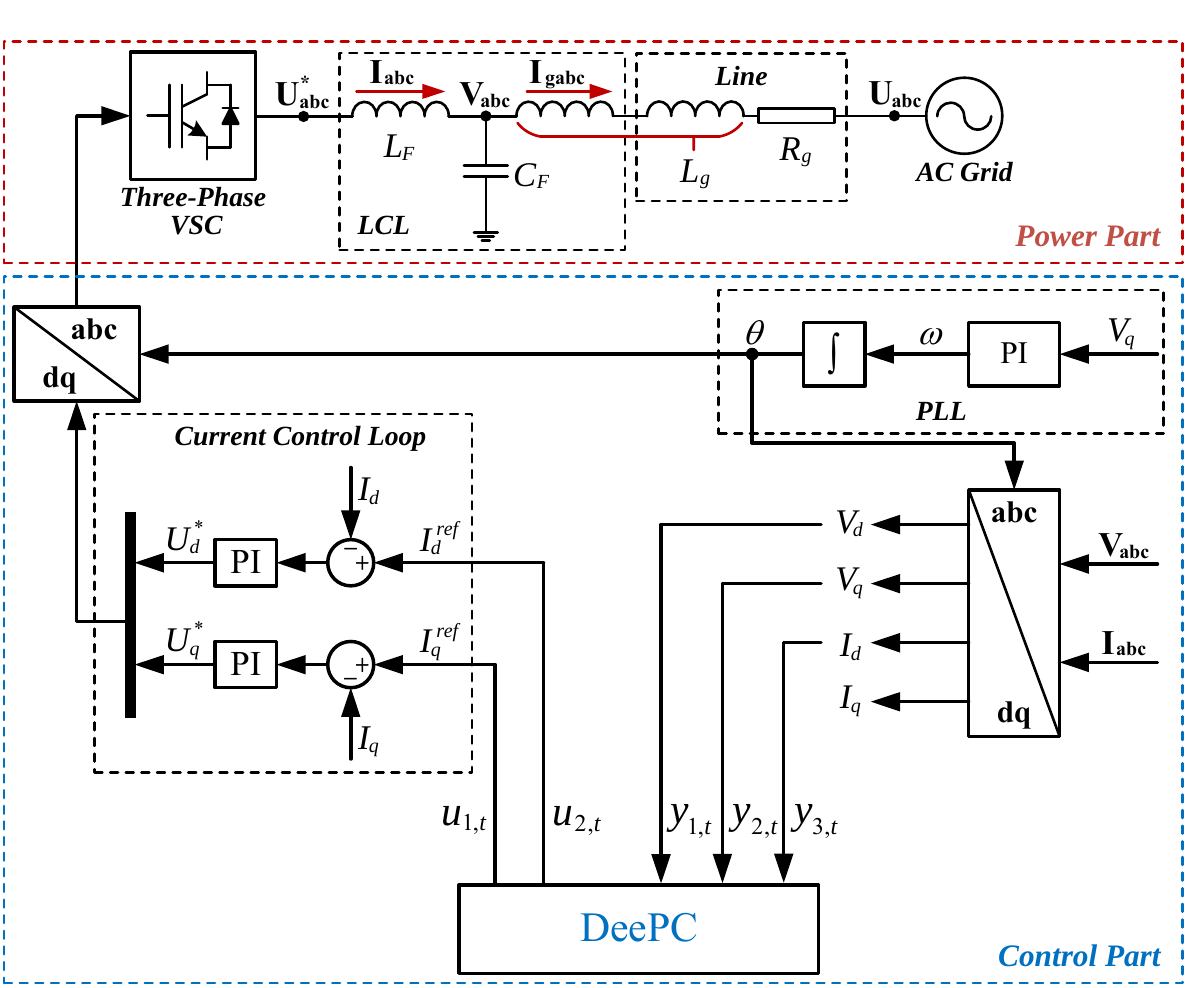}
	\caption{One-line diagram of a grid-connected power converter.}
	\label{Fig_DeePC_Converter_diagram}
\end{figure}

\renewcommand\arraystretch{1.25}
\begin{table}
	\scriptsize
	\centering
	\vspace{1em}
	\caption{Parameters of the Power Converter System}
	\begin{tabular}{|ll|}
		\hline
		\multicolumn{2}{|c|}{Base Values for Per-unit Calculation}										\\
		\hline
		Voltage base value: $U_{\rm{b}} = 380\rm{V}$	&Power base value:	$S_{\rm{b}} = 50\rm{kVA}$	\\
		Frequency base value: $f_{\rm{b}} = 50\rm{Hz}$	&												\\
		\hline
		\multicolumn{2}{|c|}{Parameters of the Power Part (per-unit values)}							\\
		\hline
		Converter-side inductor: $L_F = 0.05$			&\textit{LCL} capacitor:	$C_F = 0.05$		\\
		Grid-side inductor: $L_g = 0.35$				&Grid-side resistor: $R_g = 0.02$				\\
		\hline
		\multicolumn{2}{|c|}{Parameters of the Control Part}											\\
		\hline
		\multicolumn{2}{|l|}{PI gains of the PLL: $103.1({\rm{rad/s}}),5311.5({\rm{rad/s}})$}			\\
		\multicolumn{2}{|l|}{PI gains of the current control loop: $0.3({\rm{p.u.}}),10({\rm{p.u.}})$}	\\
		\multicolumn{2}{|l|}{Control frequency for the digital signal processor: $10\rm{kHz}$}			\\
		\hline
		\multicolumn{2}{|c|}{Parameters of the DeePC}													\\
		\hline
		\multicolumn{2}{|l|}{Length of the initial trajectories: $T_{\rm ini} = 40$}							\\
		\multicolumn{2}{|l|}{Length of the prediction horizon: $N = 30$}								\\
		\multicolumn{2}{|l|}{Length of the data to construct the Hankel matrix: $T = 500$}				\\
		\multicolumn{2}{|l|}{Sampling frequency for the DeePC: $1\rm{kHz}$}								\\
		Control cost: $R = I$							&Output cost: $Q = 400 \times I$				\\
		\multicolumn{2}{|l|}{Reference vector : $r = {\bf{1}}_N \otimes {\rm{col}}(1,0,1)$}				\\
		\multicolumn{2}{|l|}{Constraint set : $\mathcal U = \{u:-2 \times {\bf{1}}_{mN} \le u \le 2 \times {\bf{1}}_{mN}\}$}				\\
		\multicolumn{2}{|l|}{Constraint set : $\mathcal Y = \{u:r - 2 \times {\bf{1}}_{pN} \le y \le r + 2 \times {\bf{1}}_{pN}\}$}				\\
		\multicolumn{2}{|l|}{Regularization parameter: $\lambda_g = 10$}								\\
		\hline
	\end{tabular}		
	\label{table:converter_parameter}
\end{table}

Fig.\ref{Fig_DeePC_Converter_diagram} shows the one-line diagram of a three-phase power converter which is connected to an ac power grid via an \textit{LCL}. The system is nonlinear (the nonlinearity comes from the coordinate transformation), and the order of the system is $n=10$. Here, the ac power grid is modeled as an infinite bus with fixed voltage magnitude $1.0\rm{(p.u.)}$ and frequency $50\rm{Hz}$. The base values for per-unit calculation and the system parameters are given in Table \ref{table:converter_parameter}, wherein the symbol ${\bf{1}}_n \in \mathbb{R}^{n}$ denotes the column vector with all the entries being $1$, and ${\bf{0}}_n \in \mathbb{R}^{n}$ denotes the column vector with all the entries being $0$, $I$ is the identity matrix whose dimension can be inferred from the text, and the Kronecker product of $A$ and $B$ is denoted by $A \otimes B$. The control part of the converter consists of a synchronous reference frame PLL, a current control loop and coordinate transformation blocks \cite{harnefors2007modeling,rocabert2012control}.

The current control loop contains two proportional-integral (PI) regulators to make the converter-side \textit{d}-axis and \textit{q}-axis current components (i.e., $I_d$ and $I_q$ as shown in Fig.\ref{Fig_DeePC_Converter_diagram}) track their references $I_d^{ref}$ and $I_q^{ref}$, which enables fast current limiting under faults and harmonic suppression of the converter-side current.

The PLL is used for grid-synchronization. The input of the PLL is the \textit{q}-axis voltage component (i.e., $V_q$) of the \textit{LCL}'s capacitor, and the output is the angle reference $\theta$ for coordinate transformation.
The $q$-axis voltage component $V_q$ is controlled to be zero in steady state via the PI regulator in the PLL, and therefore the voltage vector is aligned with the {\em d}-axis, that is, the \textit{d}-axis voltage component (i.e., $V_d$) is the magnitude of the capacitor's voltage vector in steady state.

When the converter is connected to a strong grid that features low grid-side impedance, it will present anticipated dynamic performance. 
However, one significant challenge for the converter's operation is that generally the converter does not have any information about the properties of the ac power grid. If the converter is connected to a weak grid that has high grid-side impedance, e.g., $L_g = 0.35\rm{p.u.}$ as used in this section, the PLL will have significant interaction with the current control loop as well as the grid impedance, which may result in instabilities of the system \cite{wen2014impedance,huang2018adaptive}. This kind of small-signal instability features the oscillations of the PLL's output and current/voltage signals, which should be eliminated to ensure the safe operation of the power grid. 

In this section, we use the DeePC to perform optimal control for the converter that is connected to a weak ac grid, which can attenuate the oscillations in the converter and stabilize the system by penalizing the tracking errors in the cost function. The power converter together with the power grid is a black-box system from the view of the DeePC. The DeePC provides control inputs to this black-box system and measures its input/output trajectories with sampling frequency $1\rm{kHz}$ (i.e., sampling time $1\rm{ms}$). 

We choose $I_d^{ref}$ and $I_q^{ref}$ to be the control inputs considering that the current control loop has high response speed and allows fast current limiting. The measured outputs from the black-box system are $V_d$, $V_q$ and $I_d$, which will be controlled to track their references via the DeePC. These output signals contain measurement noise (white noise with noise power: $5.0 \times 10^{-6}$). We note that $I_q$, which corresponds to the reactive current, is not chosen as the measured output for the DeePC because $I_q$ is uniquely determined by the power flow constraint when $V_d$, $V_q$ and $I_d$ equal their reference values in steady state. Since the DeePC has no information about the black-box system, $T_{\rm ini}$ is chosen to be sufficiently large ($T_{\rm ini} = 40$) to meet $T_{\rm ini} \ge \ell$.

Fig.\ref{Fig_DeePC_Converter_curves} plots the time-domain responses of the power converter. When the DeePC is activated at $t=1.0\rm{s}$, the voltage/current oscillations are effectively eliminated. By comparison, the system is unstable without the DeePC (where $I_d^{ref}$ and $I_q^{ref}$ are respectively set as $1.0{\rm{p.u.}}$ and $0$), and there are voltage/current oscillations whose magnitudes are amplified over time. These simulation results verify the effectiveness of using DeePC to perform optimal predictive control for a grid-connected power converter with unknown grid conditions (e.g., grid-side impedance and grid properties) and unknown converter model and inner controls.

\begin{figure}[!t]
	\centering
	\includegraphics[width=2.4in]{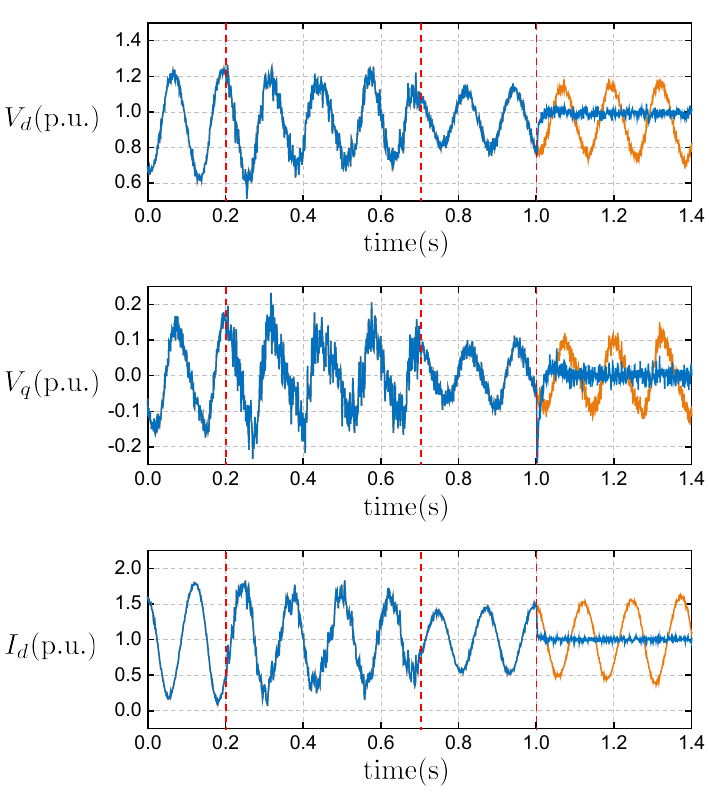}
	\caption{Time-domain responses of the power converter with DeePC. From $t=0.2\rm{s}$ to $0.7\rm{s}$, {\scriptsize{$I_d^{ref}$}} and {\scriptsize{$I_q^{ref}$}} are respectively set as $1.0{\rm{p.u.}}+ \tau_1$ and $\tau_2$ so as to get the input/output data with $u$ persistently exciting, where $\tau_1$ and $\tau_2$ are two different white noise signals (noise power: $1.0 \times 10^{-4}$). The DeePC is activated at $t=1.0\rm{s}$. {\color{ORANGE}{\bf{-----}}} without DeePC; {\color{BLUE}{\bf{-----}}} with DeePC.}
	\label{Fig_DeePC_Converter_curves}
\end{figure}


\section{From DeePC to MPC}
\subsection{Scalability of DeePC}

The DeePC approach provides a safe and optimal solution to the regulation problem by solely using the measured data from the unknown system, and thus allows to stabilize an unknown and unstable system, as illustrated in the previous case study. However, when applied in high-order systems, the DeePC may suffer from poor scalability because the dimension of the decision variable $g$ (which is $T-T_{\rm ini}-N+1$) depends on the length of data $T$ to construct the Hankel matrix. In other words, the optimization problem in (\ref{eq:DeePC}) is of high dimension when choosing a long sequence of data to possibly eliminate the impacts of measurement noise.

To remove $g$ from the constraint and ensure the scalability of the optimization problem in (\ref{eq:DeePC}), we consider a subsequent system identification and model predictive control whose decision variables are $u$ and $y$.

\subsection{Model Predictive Control}
In what follows, we consider the following finite-horizon output-based MPC problem
\begin{equation}
\begin{array}{l}
\mathop {\min }\limits_{u \in \mathcal U,y \in \mathcal Y} \;\;{\left\| u \right\|_R^2} + {\left\| {y - r} \right\|_Q^2}\\
s.t.\;\;y = K\left[ {\begin{array}{*{20}{c}}
	{{u_{\rm ini}}}\\
	{{y_{\rm ini}}}\\
	u
	\end{array}} \right] \,,	
\end{array}	
\tag{MPC}
\label{eq:MPC}
\end{equation}
where the decision variables $u$ and $y$ are the control inputs and measurement outputs over the prediction horizon, and $K \in  {\mathbb{R}^{pN \times (m{T_{\rm ini}} + p{T_{\rm ini}} + mN)}}$ is the {\em $N$-step transition matrix} predicting how future outputs of the system are determined by the initial input/output data and the future inputs. Note that the optimization problem (\ref{eq:MPC}) is solved in a receding horizon manner resulting in an online feedback control.

Since the decision variable $g$ in (\ref{eq:DeePC}) does not appear in (\ref{eq:MPC}), solving (\ref{eq:MPC}) has much less computational burden than solving (\ref{eq:DeePC}), especially when the Hankel matrix has high column number (i.e., $g$ has high dimension). On the other hand, \eqref{eq:MPC} depends on an explicit predictive model given by the transition matrix $K$ in the equality constraint.

\subsection{Prediction Error System Identification}

Observe that the predictive model in \eqref{eq:MPC} is an $N$-step ARMA model for the discrete-time LTI system mapping past inputs and outputs ${\rm col}(u_{\rm ini},y_{\rm ini})$ as well as future inputs $u$ to future outputs $y$. In particular, for the $i \rm{th}$ ($i \in \{1,...,pN\}$) element of $y$, we have
\begin{equation}
y_i = K_i \varphi\,,		\label{eq:MPC_i}
\end{equation}
where $K_i$ is the $i \rm{th}$ row of $K$ and $\varphi = {\rm{col}} (u_{\rm ini},y_{\rm ini},u)$. 

Given past measurements of $y_{i}$ and $\varphi$, the transition matrix $K_i$ can be computed offline through system identification. In the absence of measurement noise, $K_{i}$ can be computed exactly with $m{T_{\rm ini}} + p{T_{\rm ini}} + mN$ linearly independent measurements of $\varphi$ and associated $y_{i}$. 
 However, measurement noise will significantly affect the accuracy of this approach. A standard solution to remedy this problem and to eliminate the effects of the noise is to use a larger data set and apply a least-square $N$-step PEM minimizing 
\begin{equation}
\mathop {\min }\limits_{{K_i}} \;\;\sum\limits_{j = 1}^{{N_{\rm trj}}} {{{({y_{i(j)}} - {K_i}{\varphi_{(j)}})}^2}} \,,	
\tag{PEM}
\label{eq:PEM}
\end{equation}
where $N_{\rm trj}>m{T_{\rm ini}} + p{T_{\rm ini}} + mN$ is the number of measured trajectories, and $y_{i(j)}$ and $\varphi_{(j)}$ belong to the $j \rm{th}$ trajectory.
Indeed, the subsequent combination of PEM and MPC is a standard approach to model-based control that has proved itself in many applications throughout academia and industry \cite{jorgensen2011finite,camacho1999model}.

As an alternative to the batch optimization approach \eqref{eq:PEM} combining all the measured trajectories to solve for $K$ in one step, $K$ can be obtained by adopting the recursive least-square algorithm \cite{ljung1998system}
\begin{equation}
\left\{ {\begin{array}{*{20}{l}}
	{{K_{i(j)}} = {K_{i(j - 1)}} + L_{(j)}^T\left( {{y_{i(j)}} - K_{i(j - 1)}{\varphi_{(j)}}} \right)}\\
	{{L_{\left( j \right)}} = \frac{{{P_{\left( {j - 1} \right)}}{\varphi_{(j)}}}}{{1 + \varphi_{(j)}^T{P_{\left( {j - 1} \right)}}{\varphi_{(j)}}}}}\\
	{{P_{\left( j \right)}} = {P_{\left( {j - 1} \right)}} - \frac{{{P_{\left( {j - 1} \right)}}{\varphi_{(j)}}\varphi_{(j)}^T{P_{\left( {j - 1} \right)}}}}{{1 + \varphi_{(j)}^T{P_{\left( {j - 1} \right)}}{\varphi_{(j)}}}}}
	\end{array}} \right.\,,		\label{eq:MPC_rls}
\end{equation}
where $K_{i(j-1)}$ is the least-square solution for $K_i$ after combining the previous $j-1$ trajectories, and $K_{i(j-1)}$ is updated to $K_{i(j)}$ using the data (i.e., $y_{i(j)}$ and $\varphi_{(j)}$) in the $j\rm{th}$ trajectory. The matrices $L_{(j)} \in \mathbb{R}^{(m{T_{\rm ini}} + p{T_{\rm ini}} + mN) \times 1}$ and $P_{(j)} \in \mathbb{R}^{(m{T_{\rm ini}} + p{T_{\rm ini}} + mN) \times (m{T_{\rm ini}} + p{T_{\rm ini}} + mN)}$ are two intermediate variables updated through the $j\rm{th}$ trajectory in the recursive algorithm. 

The recursive algorithm \eqref{eq:MPC_rls} is not only more scalable for large data sets, but it could possibly also be applied online to adapt the model $K$ used in \eqref{eq:MPC} to a variable environment. Moreover, the size of $K$ is solely related to $T_{\rm ini}$ and $N$ and independent of the size of the data set, which leads to a smaller optimization problem size than (\ref{eq:DeePC}) when applied to high-order systems, considering that the dimension of the decision variable $g$ in (\ref{eq:DeePC}) depends on the length of the data $T$ to construct the Hankel matrix.

\subsection{Relation of DeePC, MPC, and PEM}

In the following, we will relate the subsequent system identification through \eqref{eq:PEM} and model-based control through \eqref{eq:MPC} to the data-driven \eqref{eq:DeePC} strategy.
In a first step, observe that the least-square solution for (\ref{eq:PEM}) can be expressed in closed form as
\begin{equation}
K_i = [y_{i(1)}\;\; \cdots \;\;y_{i(N_{\rm trj})}][\varphi_{(1)}\;\; \cdots \;\;\varphi_{(N_{\rm trj})}]^+\,, \label{eq:ls_Ki}
\end{equation}
where the Moore-Penrose pseudoinverse of the matrix $X$ is denoted by $X^+$.

The transition matrix $K$ can then be obtained as
\begin{equation}
K = [y_{(1)}\;\; \cdots \;\;y_{(N_{\rm trj})}][\varphi_{(1)}\;\; \cdots \;\;\varphi_{(N_{\rm trj})}]^+\,.		\label{eq:ls_K}
\end{equation}

Since every column in the constructed Hankel matrix (\ref{eq:partition_Huy}) is one trajectory measured from the unknown system, $K$ can also be calculated by using the Hankel matrix (\ref{eq:partition_Huy}) by
\begin{equation}
K = {Y_f}{\left[ {\begin{array}{*{20}{c}}
		{{U_P}}\\
		{{Y_P}}\\
		{{U_f}}
		\end{array}} \right]^ + }\,.		\label{eq:ls_K_Hankel}
\end{equation}

By combining (\ref{eq:ls_K_Hankel}) with $y=K\varphi$ and $y=Y_f g$ one obtains
\begin{equation}
g = {\left[ {\begin{array}{*{20}{c}}
		{{U_P}}\\
		{{Y_P}}\\
		{{U_f}}
		\end{array}} \right]^ + }\left[ {\begin{array}{*{20}{c}}
	{{u_{\rm ini}}}\\
	{{y_{\rm ini}}}\\
	u
	\end{array}} \right]\,,		\label{eq:g_n2}
\end{equation}
which equals the solution of the optimization problem
\begin{equation}
\begin{array}{l}
\mathop {\min }\limits_g \;\;{\left\| g \right\|_2^2}\\
s.t.\;\;\left[ {\begin{array}{*{20}{c}}
	{{U_P}}\\
	{{Y_P}}\\
	{{U_f}}
	\end{array}} \right]g = \left[ {\begin{array}{*{20}{c}}
	{{u_{\rm ini}}}\\
	{{y_{\rm ini}}}\\
	u
	\end{array}} \right]\,.
\end{array}
\tag{LN}
\label{eq:LN}
\end{equation}
Observe that the optimization problem \eqref{eq:LN} seeks the least-norm solution to the equality constraint. Thus, we refer to it as the {\em least-norm} ({\bf LN}) problem. As derived above, for Hankel matrix data, the  solution $g'^{\star}$ of \eqref{eq:LN} is related to the solution of the least-squares \eqref{eq:PEM} problem $K^\star$ by $y = Y_{f}g'^{\star} = K^\star \varphi$. It can also be derived using subspace identification methods \cite{markovsky2005algorithms}. 
We summarize this observation below.

\begin{lemma}[PEM and LN]
\label{Lemma: PEM and LN}
Consider the least-norm problem \eqref{eq:LN} with Hankel matrix ${\rm{col}}(U_P,Y_P,U_f)$. Consider the least-square $N$-step prediction error method optimization problem \eqref{eq:PEM}, and assume that its data $y_{i(j)}$ and $\varphi_{(j)}$ is arranged into the Hankel matrix ${\rm{col}}(U_P,Y_P,U_f,Y_f)$. 
Then the solution $K^{\star}$ of \eqref{eq:PEM} and the solution $g'^{\star}$ of \eqref{eq:LN} are related by the equation $K^{\star}\varphi=Y_{f}g'^{\star}$.
\end{lemma}

Observe that if the \eqref{eq:LN} solution $g'^\star$ is used as predictive model (equality constraint) of the \eqref{eq:MPC} problem by setting $y = Y_{f}g'^{\star}$ (where $g'^{\star}$ is a function of $u$), then the subsequent concatenation of the \eqref{eq:MPC} and the \eqref{eq:LN} (or equivalently \eqref{eq:PEM} by Lemma \ref{Lemma: PEM and LN}) optimization problems reads as follows:
\begin{equation}
\begin{array}{l}
\mathop {\min }\limits_{u' \in \mathcal U,y' \in \mathcal Y} \;\;{\left\| u' \right\|_R^2} + {\left\| {y' - r} \right\|_Q^2}\\
s.t.\;\;u' = {U_f}g'^\star,\;y' = {Y_f}g'^\star,\\
{\rm{where}}\;\;g'^\star = \mathop {\arg \min }\limits_{g'} \;\;{\left\| g' \right\|_2^2}\\
\;\;\;\;\;\;\;\;\;\;\;s.t.\;\;\left[ {\begin{array}{*{20}{c}}
	{{U_P}}\\
	{{Y_P}}\\
	{{U_f}}
	\end{array}} \right]g' = \left[ {\begin{array}{*{20}{c}}
	{{u_{\rm ini}}}\\
	{{y_{\rm ini}}}\\
	u'
	\end{array}} \right]\,.		
\end{array}
\tag{PEM-MPC}
\label{eq:PEM-MPC}
\end{equation}

Note that $g'^\star$ depends on the decision variable $u'$ in \eqref{eq:PEM-MPC}. The inner problem of \eqref{eq:PEM-MPC} is the system identification step through \eqref{eq:LN} (or equivalently \eqref{eq:PEM} by Lemma \ref{Lemma: PEM and LN}).  The outer problem on the other hand is identical to the \eqref{eq:MPC} problem. Let 
\begin{equation}
C(u,y,g) = {\left\| {{u}} \right\|_R^2} + {\left\| {{y} - r} \right\|_Q^2} + {\lambda _g}{\left\| {{g}} \right\|_2^2}
\label{eq:CombinedCost}
\end{equation}
be the combined cost taking into account both system performance (i.e., ${\left\| {{u}} \right\|_R^2} + {\left\| {{y} - r} \right\|_Q^2}$) and the complexity of $g$ (i.e., ${\lambda _g}{\left\| {{g}} \right\|_2^2}$). This is in fact the true cost we wish to minimize when noise is present in the system as large entries of $g$ result in overfitting of the noisy trajectories in the Hankel matrix in \eqref{eq:Hankel_g} \cite{coulson2018data}. Since this cost appears in \eqref{eq:DeePC}, we can compare directly the performance of \eqref{eq:PEM-MPC} and \eqref{eq:DeePC} with respect to this cost. We present the comparison below.


\begin{lemma}[PEM-MPC and DeePC]
\label{Lemma: PEM-MPC and DeePC}
Consider the optimal solution $(u'^\star,y'^\star,g'^\star)$ of the concatenated \eqref{eq:PEM-MPC} problem and the optimal solution $(u^\star,y^\star,g^\star)$ of the \eqref{eq:DeePC} problem. It holds that 
\begin{equation*}
C(u^{\star},y^{\star},g^{\star}) \le C(u'^{\star},y'^{\star},g'^{\star})	\,.
\end{equation*}
That is, \eqref{eq:DeePC} achieves a cost less or equal to the cost of the concatenated \eqref{eq:PEM-MPC} with respect to \eqref{eq:CombinedCost}.
\end{lemma}

\begin{proof}
Observe that any feasible point $(u',y',g'^\star)$ of~\eqref{eq:PEM-MPC} is also a feasible point of~\eqref{eq:DeePC}. Since $g'^\star$ is a particular $g$ that satisfies the constraints of~\eqref{eq:DeePC} then the feasible set of~\eqref{eq:PEM-MPC} is a subset of~\eqref{eq:DeePC}.  
Hence, it holds that $C(u^{\star},y^{\star},g^{\star}) \le C(u'^{\star},y'^{\star},g'^{\star})$.
\end{proof}

In other words, the DeePC presents better performance than the MPC formulated in (\ref{eq:MPC}) if $K$ is obtained by (\ref{eq:PEM}). On the other hand, the MPC problem in (\ref{eq:MPC}) has the advantage that it doesn't contain the decision variable $g$ and therefore has lower computational burden, which make it scalable to high-order systems. Furthermore, a recursive algorithm can be used in PEM-MPC to obtain the $N$-step transition matrix $K$ online, which enables the implementation for digital signal processors. In the next section we numerically observe that the gap between (\ref{eq:DeePC}) and (\ref{eq:PEM-MPC}) presented in Lemma \ref{Lemma: PEM-MPC and DeePC} can actually be made smaller by an appropriate choice of regularization $\lambda_g$.

\subsection{Comparative Case Studies} 

Fig.\ref{Fig_Comparative_results} displays the time-domain responses of the power converter when the PEM-MPC is used, with the other settings remaining the same as those in Fig.\ref{Fig_DeePC_Converter_diagram}. Before $t=0\rm{s}$, $I_d^{ref}$ and $I_q^{ref}$ were respectively set as $1.0{\rm{p.u.}}+ \tau_1$ and $\tau_2$ for $30\rm{s}$, where $\tau_1$ and $\tau_2$ are two different white noise signals (noise power: $1.0 \times 10^{-4}$). During this process, the power converter is critically stable with a grid-side inductor $L_g = 0.34\rm{p.u.}$, and the other parameters are the same as those in Table \ref{table:converter_parameter}.

\begin{figure}[!t]
	\centering
	\includegraphics[width=2.5in]{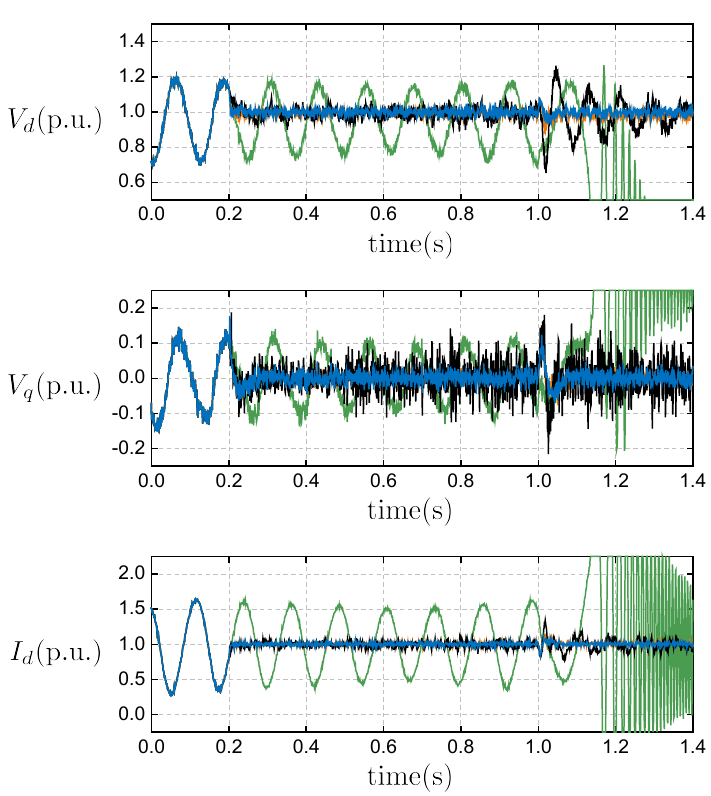}
	\caption{Time-domain responses of the power converter with different algorithms. The DeePC/PEM-MPC is activated at $t=0.2\rm{s}$. The grid-side inductance $L_g$ is changed from $0.34\rm{p.u.}$ to $0.35\rm{p.u.}$ at $t=0.7\rm{s}$, and to $0.5\rm{p.u.}$ at $1.0\rm{s}$. 
		{\color{BLUE}{\bf{-----}}} PEM-MPC; 
		{\color{ORANGE}{\bf{-----}}} DeePC ($T=500$);
		{\color{black}{\bf{-----}}} DeePC ($T=330$);
		{\color{GREEN1}{\bf{-----}}} {\scriptsize{$I_d^{ref}=1.0{\rm{p.u.}}, I_q^{ref}=0$}}.}
	\label{Fig_Comparative_results}
\end{figure}

By using recursive least-square algorithm, the $N$-step transition matrix $K$ was calculated iteratively with measurements spaced by $20\rm{ms}$, that is, $K$ has been obtained during the converter's operation by recursively solving the PEM problem in (\ref{eq:PEM}) with $1500$ trajectories before $t=0\rm{s}$.

At $t=0.2\rm{s}$, the PEM-MPC is activated, which effectively attenuates the voltage and current oscillations and stabilizes the system. Then, we test the robustness of the PEM-MPC by changing the grid-side inductance $L_g$ from $0.34\rm{p.u.}$ to $0.35\rm{p.u.}$ at $t=0.7\rm{s}$, and to $0.5\rm{p.u.}$ at $1.0\rm{s}$. It can be seen that the system is stable under these two disturbances, and the current and voltage signals track the references with fast dynamics, that is, the PEM-MPC shows robustness in terms of parameter changes in the black-box system. Here $K$ is not updated during the two disturbances to test how an inaccurate model in $K$ affects the performance of the PEM-MPC.

For comparison, Fig.\ref{Fig_Comparative_results} also plots the time-domain responses of the power converter when the DeePC is applied. By choosing $T=500$, the DeePC also effectively eliminates the voltage and current oscillations when activated at $t=0.2\rm{s}$, and presents robustness when $L_g$ is changed from $0.34\rm{p.u.}$ to $0.35\rm{p.u.}$ at $t=0.7\rm{s}$, and to $0.5\rm{p.u.}$ at $1.0\rm{s}$. Note the Hankel matrix is obtained from the system with $L_g=0.34\rm{p.u.}$ and is not updated during the disturbances.

When choosing $T=330$ in the DeePC, the oscillations are obviously eliminated as well after the DeePC is activated. However, the voltage and current signals are oscillating after $L_g$ is changed from $0.35\rm{p.u.}$ to $0.5\rm{p.u.}$ at $1.0\rm{s}$. This is because the trajectory ${\rm{col}}(u^{\rm{d}},y^{\rm{d}})$ in this case contains less information of the system than that with $T=500$, and thus the prediction is more strongly affected by the measurement noise when solving the DeePC. Choosing a larger $T$ is a convenient way to resolve this problem, yet it results in a higher dimension of $g$ and thus higher computational burden in solving (\ref{eq:DeePC}), which may not be scalable to high-$\ell$ systems as discussed before. 

Fig.\ref{Fig_Cost_T} plots the relationship between $T$ and the time-domain cost (from $t=0.2\rm{s}$ to $t=1.4\rm{s}$) to illustrate how different values of $T$ affect the system performance, where the time-domain cost is $\left\| {{u_{\rm time}}} \right\|_R^2 + \left\| {{y_{\rm time}} - {r_{\rm time}}} \right\|_Q^2$ ($u_{\rm time}$ and $y_{\rm time}$ are the input and output trajectories measured from the system, and $r_{\rm time}$ is the reference vector for the outputs). It is shown that the time-domain cost sharply decreases with the increase of $T$ from $320$ to $400$, and remain nearly the same (or slightly increases) if further increasing $T$. 

\begin{figure}[!t]
	\centering
	\includegraphics[width=2.1in]{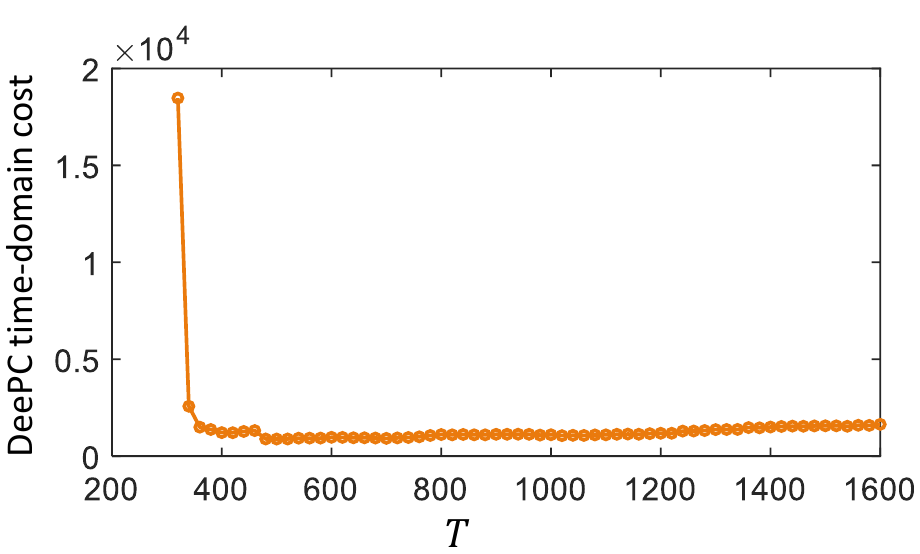}
	\caption{Variations of the time-domain cost with different values of $T$.}
	\label{Fig_Cost_T}
\end{figure}

For comparison, the converter's responses without the DeePC or the PEM-MPC are plotted by the green lines in Fig.\ref{Fig_Comparative_results}. In this case, the grid-side inductance $L_g$ is also changed from $0.34\rm{p.u.}$ to $0.35\rm{p.u.}$ at $t=0.7\rm{s}$, and to $0.5\rm{p.u.}$ at $1.0\rm{s}$. It can be seen that the system is unstable with the voltage and current signals keeping oscillating, which endangers the power system operation.

\begin{figure}[!t]
	\centering
	\includegraphics[width=2.8in]{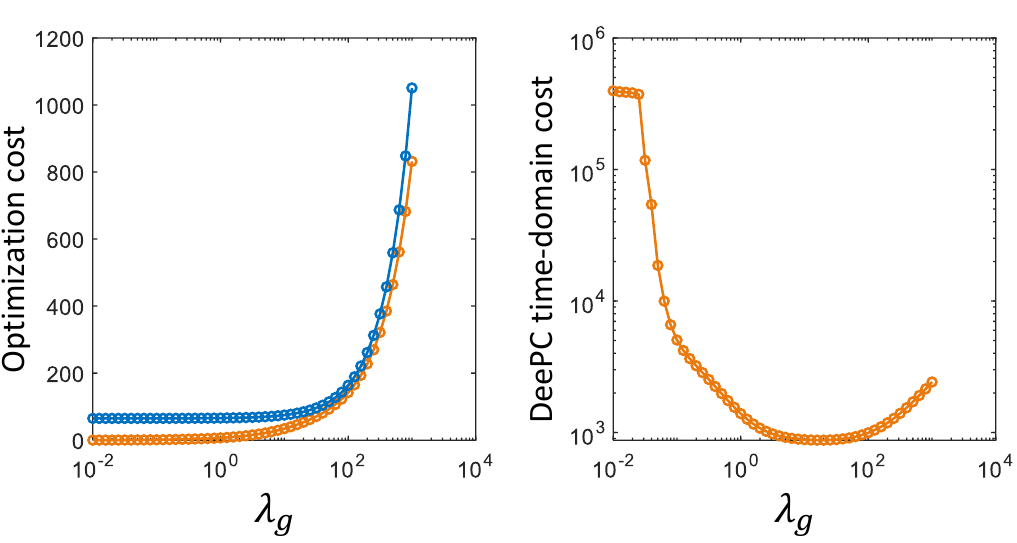}
	\caption{Variations of the optimization cost and the time-domain cost with different values of $\lambda_g$. {\color{ORANGE}{\bf{-----}}} DeePC; {\color{BLUE}{\bf{-----}}} PEM-MPC;}
	\label{Fig_Cost_lambda_g}
\end{figure}

Fig.\ref{Fig_Cost_lambda_g} displays how $\lambda_g$ affects the optimization cost of the system (setting $T=500$). The optimization costs of DeePC and PEM-MPC are respectively obtained by solving (\ref{eq:DeePC}) and (\ref{eq:PEM-MPC}) at $t=0.2\rm{s}$, and Fig.\ref{Fig_Cost_lambda_g} shows that DeePC outperforms PEM-MPC under any $\lambda_g$, which is consistent with Lemma \ref{Lemma: PEM-MPC and DeePC}. In addition, Fig.\ref{Fig_Cost_lambda_g} also plots how  $\lambda_g$ affects the time-domain cost (from $t=0.2\rm{s}$ to $t=1.4\rm{s}$) of the DeePC, and shows that the time-domain cost is dramatically reduced when increasing $\lambda_g$ from $0.01$ to $20$, and it slightly increases if further increasing $\lambda_g$.

\section{Application of PEM-MPC in Large-Scale Systems}

To illustrate the effectiveness of the PEM-MPC in large-scale systems, we now provide a detailed simulation study based on a {\em{nonlinear}} model of the three-phase two-area test system (the order of the system is $n=90$) in Fig.\ref{Fig_Two_Area_System}. This test system consists of four SGs and one VSC-HVDC station. The VSC-HVDC station is a large-capacity power converter with an {\em{LCL}} filter, which in our case, employs the control scheme given in Fig.\ref{Fig_DeePC_Converter_diagram}. The base values for this test system are: $f_{\rm{b}}=50\rm{Hz}$, $S_{\rm{b}}=350\rm{MVA}$ and $U_{\rm{b}}=220\rm{kV}$ (line to line). In the VSC-HVDC station, the grid-side inductor of the {\em{LCL}} filter is chosen as $L_g = 0.05{\rm{p.u.}}$, and the other parameters are the same as those in Table \ref{table:converter_parameter} unless otherwise specified. The parameters of the SGs and the power grid are chosen as the same as those in \cite{huang2018damping}. 

\begin{figure}[!t]
	\centering
	\includegraphics[width=3in]{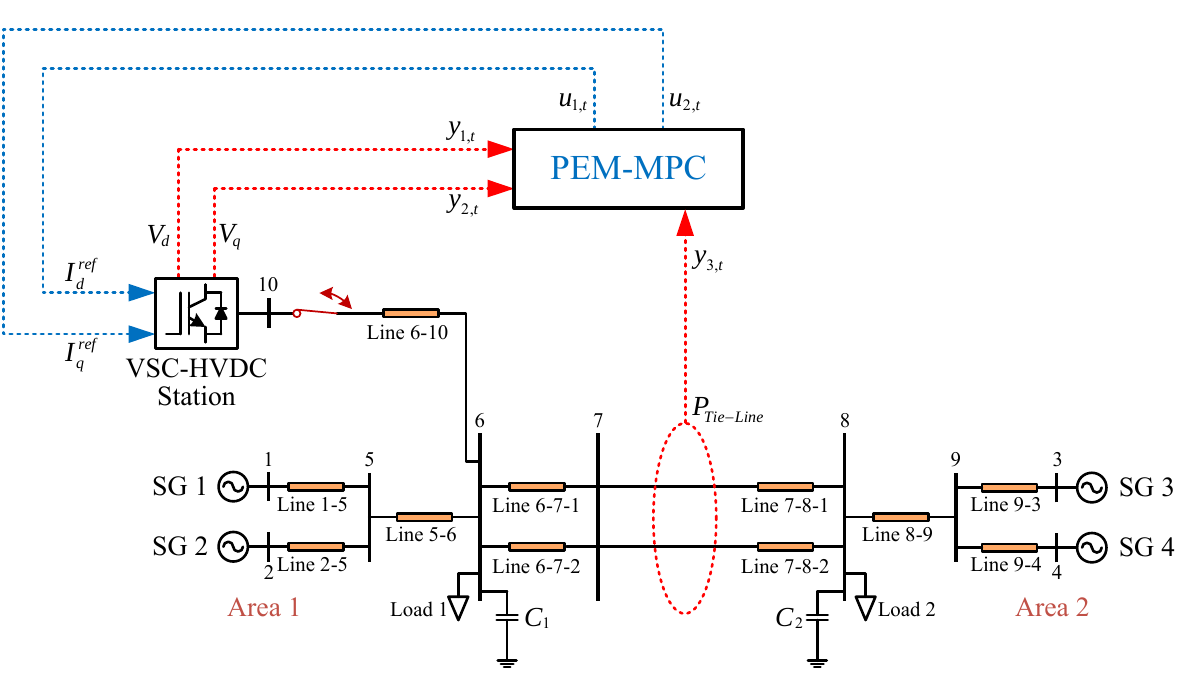}
	\caption{One-line diagram of a three-phase two-area test system with integration of a VSC-HVDC station.}
	\label{Fig_Two_Area_System}
\end{figure}

This test system has one pair of weakly-damped inter-area modes and presents low-frequency oscillations, as caused by the fast exciters in the SGs as well as long transmission lines \cite{kundur1994power}. The PEM-MPC enables us to use the high controllability of the converter to attenuate the inter-area oscillations. We still choose $I_d^{ref}$ and $I_q^{ref}$ to be control inputs provided by the PEM-MPC, and $V_d$ and $V_q$ are two of the three measured outputs from the converter. Another output signal is the inter tie-line active power $P_{Tie-Line}$ transferred from Area 1 to Area 2, as shown in Fig.\ref{Fig_Two_Area_System}. All the outputs contain measurement noise (white noise with noise power $1.0 \times 5.0^{-6}$). We note that DeePC may not be suitable to be applied to such a high-order system because (\ref{eq:DeePC}) contains a high-dimension decision variable $g$ and thus cannot be solved in real time.

Fig.\ref{Fig_Two_Area_System_curves} plots the time-domain responses of the inter tie-line active power when different $T_{\rm ini}$ and $N$ are adopted. The sampling time for PEM-MPC is $1\rm{ms}$. Note that before $t=0\rm{s}$, white noise signals (noise power: $1.0 \times 10^{-4}$) were injected into the test system through $I_d^{ref}$ and $I_q^{ref}$, which lasted for 1 minute. During this process, the $N$-step transition matrix $K$ was calculated iteratively with measurements spaced by $20\rm{ms}$, that is, $K$ was obtained by solving the PEM problem with 3000 trajectories but in the following cases, $K$ is not adapted online when the PEM-MPC is performed.

Firstly, since the test system is of high order (with $90$ state variables) and is a black-box system to the PEM-MPC, we choose a sufficiently large $T_{\rm ini}$ ($T_{\rm ini}=200$) to meet $T_{\rm ini} \ge \ell$, and the prediction horizon is chosen as $N=80$. It can be seen that the inter-area oscillation is well eliminated after the PEM-MPC is activated at $t=10\rm{s}$. Then, considering that the low-frequency oscillations are caused by the interactions among the SGs' rotors and thus can possibly be represented by an equivalent low-order system, $T_{\rm ini}$ and $N$ are chosen smaller to test the performance of the PEM-MPC. Fig.\ref{Fig_Two_Area_System_curves} shows that when $(T_{\rm ini},N)$ are chosen to be $(10,10)$ and even $(5,10)$, the PEM-MPC can still effectively attenuate the inter tie-line power oscillations. Also observe that with the decrease of $T_{\rm ini}$ and $N$, some slight oscillations still exist after the PEM-MPC is activated, which is caused by the model mismatching, i.e., a smaller size of $K$ may not be able to accurately represent the full-order model.

\begin{figure}[!t]
	\centering
	\includegraphics[width=2.5in]{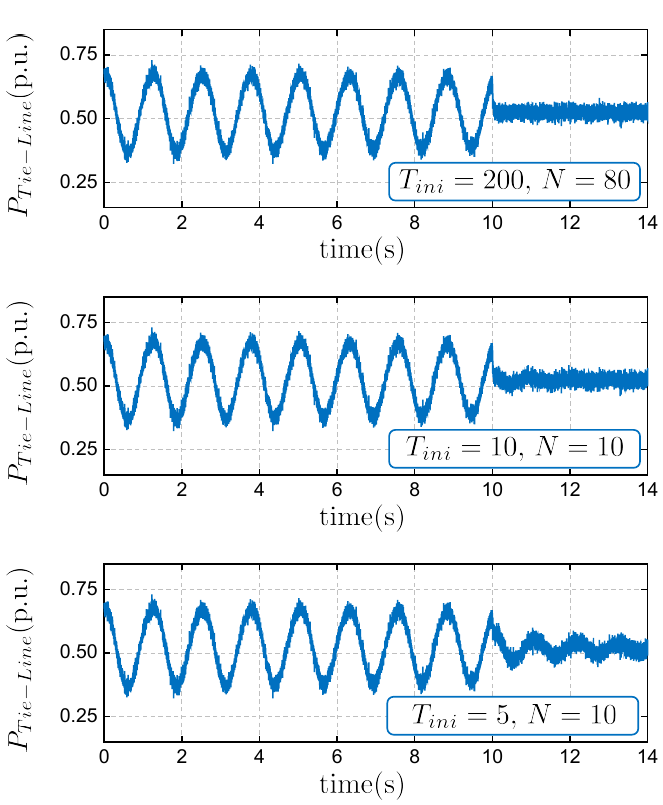}
	\caption{Time-domain responses of the inter tie-line active power. The PEM-MPC is activated at $t=10\rm{s}$.}
	\label{Fig_Two_Area_System_curves}
\end{figure}

To further illustrate how the choice of $T_{\rm ini}$ affects the overall performance, Fig.\ref{Fig_Cost_Tini} plots the variations of the time-domain cost with different values of $T_{\rm ini}$ (from $t=10\rm{s}$ to $t=14\rm{s}$). Fig.\ref{Fig_Cost_Tini} shows that the time-domain cost is significantly reduced with the increase of $T_{\rm ini}$ from $5$ to $50$, but remains nearly the same if further increasing $T_{\rm ini}$, which indicates the test system is well represented with $T_{\rm ini}=50$.

\begin{figure}[!t]
	\centering
	\includegraphics[width=2in]{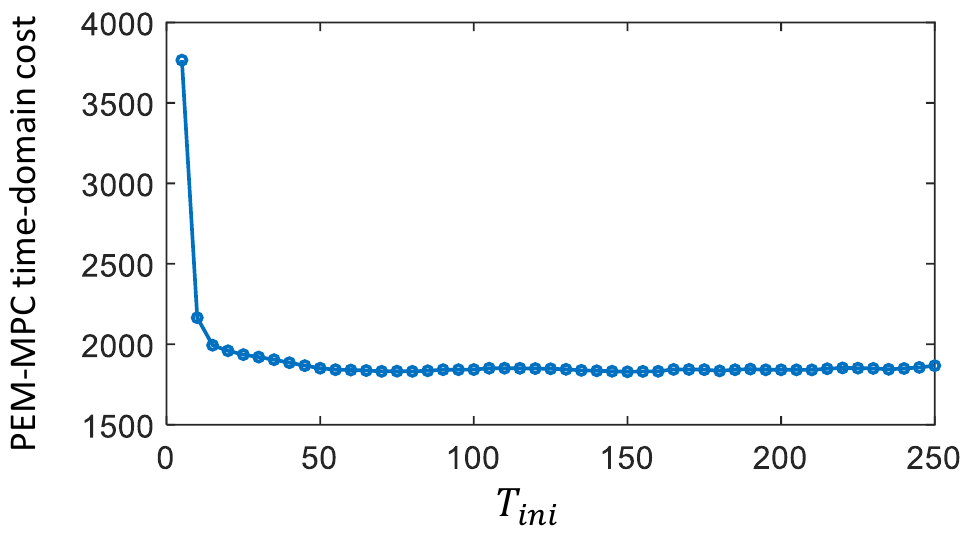}
	\caption{Variations of the time-domain cost with different $T_{\rm ini}$.}
	\label{Fig_Cost_Tini}
\end{figure}

\section{Conclusion}

We applied DeePC in grid-connected power converters to eliminate undesired oscillations caused by weak grid conditions. The DeePC has no information about the power converter or the power grid, and solely uses input/output data measured from the unknown system. We showed that the DeePC can stabilize an unstable system by performing optimal and constrained receding-horizon control. Moreover, the DeePC presents more robustness against changes in the unknown system with a longer input/output data sequence to construct the Hankel matrix, but meanwhile results in higher computational burden and poor scalability especially in high-order systems. For this reason, we presented a concatenated PEM-MPC method as an alternative model-based, optimal, and constrained receding-horizon control, wherein the PEM can be implemented in a recursive manner. We discussed the connection between DeePC and PEM-MPC, and formally showed that the DeePC outperforms PEM-MPC with regards to the cost in the optimization problem. We applied the PEM-MPC in a VSC-HVDC station which is connected to a high-order power grid that contains four SGs, and showed that the PEM-MPC effectively eliminates the low-frequency oscillations caused by the interactions among multiple SGs.

\bibliographystyle{IEEEtran}

\bibliography{ref}

\begin{thebibliography}{10}
\providecommand{\url}[1]{#1}
\csname url@samestyle\endcsname
\providecommand{\newblock}{\relax}
\providecommand{\bibinfo}[2]{#2}
\providecommand{\BIBentrySTDinterwordspacing}{\spaceskip=0pt\relax}
\providecommand{\BIBentryALTinterwordstretchfactor}{4}
\providecommand{\BIBentryALTinterwordspacing}{\spaceskip=\fontdimen2\font plus
\BIBentryALTinterwordstretchfactor\fontdimen3\font minus
  \fontdimen4\font\relax}
\providecommand{\BIBforeignlanguage}[2]{{%
\expandafter\ifx\csname l@#1\endcsname\relax
\typeout{** WARNING: IEEEtran.bst: No hyphenation pattern has been}%
\typeout{** loaded for the language `#1'. Using the pattern for}%
\typeout{** the default language instead.}%
\else
\language=\csname l@#1\endcsname
\fi
#2}}
\providecommand{\BIBdecl}{\relax}
\BIBdecl

\bibitem{FM-FD-GH-DH-GV:18}
F.~Milano, F.~D{\"o}rfler, G.~Hug, D.~Hill, and G.~Verbic, ``Foundations and
  challenges of low-inertia systems,'' in \emph{Power Systems Computation
  Conference (PSCC)}, Dublin, Ireland, 2018.

\bibitem{olivares2014trends}
D.~E. Olivares, A.~Mehrizi-Sani, A.~H. Etemadi \emph{et~al.}, ``Trends in
  microgrid control,'' \emph{IEEE Trans. Smart Grid}, vol.~5, no.~4, pp.
  1905--1919, 2014.

\bibitem{pogaku2007modeling}
N.~Pogaku, M.~Prodanovic, and T.~C. Green, ``Modeling, analysis and testing of
  autonomous operation of an inverter-based microgrid,'' \emph{IEEE Trans.
  Power Electron.}, vol.~22, no.~2, pp. 613--625, 2007.

\bibitem{harnefors2007modeling}
L.~Harnefors, ``Modeling of three-phase dynamic systems using complex transfer
  functions and transfer matrices,'' \emph{IEEE Trans. Ind. Electron.},
  vol.~54, no.~4, pp. 2239--2248, 2007.

\bibitem{cespedes2014impedance}
M.~Cespedes and J.~Sun, ``Impedance modeling and analysis of grid-connected
  voltage-source converters,'' \emph{IEEE Trans. Power Electron.}, vol.~29,
  no.~3, pp. 1254--1261, 2014.

\bibitem{wen2016analysis}
B.~Wen, D.~Boroyevich, R.~Burgos, P.~Mattavelli, and Z.~Shen, ``Analysis of dq
  small-signal impedance of grid-tied inverters,'' \emph{IEEE Trans. Power
  Electron.}, vol.~31, no.~1, pp. 675--687, 2016.

\bibitem{suul2016impedance}
J.~A. Suul, S.~D'Arco, P.~Rodr{\'\i}guez, and M.~Molinas,
  ``Impedance-compensated grid synchronisation for extending the stability
  range of weak grids with voltage source converters,'' \emph{IET Generation,
  Transmission \& Distribution}, vol.~10, no.~6, pp. 1315--1326, 2016.

\bibitem{huang2018adaptive}
L.~Huang, H.~Xin, Z.~Wang \emph{et~al.}, ``An adaptive phase-locked loop to
  improve stability of voltage source converters in weak grids,'' in
  \emph{Power and Energy Society General Meeting (PESGM), Portland, OR,
  USA}.\hskip 1em plus 0.5em minus 0.4em\relax IEEE, 2018, pp. 1--5.

\bibitem{weiss2004h}
G.~Weiss, Q.-C. Zhong, T.~C. Green, and J.~Liang, ``H/sup/spl infin//repetitive
  control of dc-ac converters in microgrids,'' \emph{IEEE Transactions on Power
  Electronics}, vol.~19, no.~1, pp. 219--230, 2004.

\bibitem{lewis2012reinforcement}
F.~L. Lewis, D.~Vrabie, and K.~G. Vamvoudakis, ``Reinforcement learning and
  feedback control: Using natural decision methods to design optimal adaptive
  controllers,'' \emph{IEEE Control Systems Magazine}, vol.~32, no.~6, pp.
  76--105, 2012.

\bibitem{dean2017sample}
S.~Dean, H.~Mania, N.~Matni, B.~Recht, and S.~Tu, ``On the sample complexity of
  the linear quadratic regulator,'' \emph{arXiv preprint arXiv:1710.01688},
  2017.

\bibitem{boczar2018finite}
R.~Boczar, N.~Matni, and B.~Recht, ``Finite-data performance guarantees for the
  output-feedback control of an unknown system,'' in \emph{2018 IEEE Conference
  on Decision and Control (CDC)}.\hskip 1em plus 0.5em minus 0.4em\relax IEEE,
  2018, pp. 2994--2999.

\bibitem{coulson2018data}
J.~Coulson, J.~Lygeros, and F.~D{\"o}rfler, ``Data-enabled predictive control:
  In the shallows of the deepc,'' \emph{arXiv preprint arXiv:1811.05890}, 2018.

\bibitem{markovsky2006exact}
I.~Markovsky, J.~C. Willems, S.~Van~Huffel, and B.~De~Moor, \emph{Exact and
  approximate modeling of linear systems: A behavioral approach}.\hskip 1em
  plus 0.5em minus 0.4em\relax SIAM, 2006, vol.~11.

\bibitem{willems2005note}
J.~C. Willems, P.~Rapisarda, I.~Markovsky, and B.~L. De~Moor, ``A note on
  persistency of excitation,'' \emph{Systems \& Control Letters}, vol.~54,
  no.~4, pp. 325--329, 2005.

\bibitem{markovsky2008data}
I.~Markovsky and P.~Rapisarda, ``Data-driven simulation and control,''
  \emph{International Journal of Control}, vol.~81, no.~12, pp. 1946--1959,
  2008.

\bibitem{markovsky2005algorithms}
I.~Markovsky, J.~C. Willems, P.~Rapisarda, and B.~L. De~Moor, ``Algorithms for
  deterministic balanced subspace identification,'' \emph{Automatica}, vol.~41,
  no.~5, pp. 755--766, 2005.

\bibitem{bjork2019performance}
J.~Bj{\"o}rk, ``Performance quantification of interarea oscillation damping
  using hvdc,'' Ph.D. dissertation, KTH Royal Institute of Technology, 2019.

\bibitem{huang2018damping}
L.~Huang, H.~Xin, and Z.~Wang, ``Damping low-frequency oscillations through
  vsc-hvdc stations operated as virtual synchronous machines,'' \emph{IEEE
  Trans. Power Electron.}, 2018, {early access.}

\bibitem{kundur1994power}
P.~Kundur, N.~J. Balu, and M.~G. Lauby, \emph{Power system stability and
  control}.\hskip 1em plus 0.5em minus 0.4em\relax McGraw-hill New York, 1994,
  vol.~7.

\bibitem{Coulson2019Regularized}
J.~Coulson, J.~Lygeros, and F.~D{\"o}rfler, ``Regularized and distributionally
  robust data-enabled predictive control,'' in \emph{2019 IEEE Conference on
  Decision and Control (CDC)}, 2019, submitted.

\bibitem{rocabert2012control}
J.~Rocabert, A.~Luna, F.~Blaabjerg, and P.~Rodriguez, ``Control of power
  converters in ac microgrids,'' \emph{IEEE Trans. Power Electron.}, vol.~27,
  no.~11, pp. 4734--4749, 2012.

\bibitem{wen2014impedance}
B.~Wen, D.~Boroyevich, P.~Mattavelli \emph{et~al.}, ``Impedance-based analysis
  of grid-synchronization stability for three-phase paralleled converters,'' in
  \emph{Applied Power Electronics Conference and Exposition (APEC)}.\hskip 1em
  plus 0.5em minus 0.4em\relax IEEE, 2014, pp. 1233--1239.

\bibitem{jorgensen2011finite}
J.~B. J{\o}rgensen, J.~K. Huusom, and J.~B. Rawlings, ``Finite horizon mpc for
  systems in innovation form,'' in \emph{2011 50th IEEE Conference on Decision
  and Control and European Control Conference}.\hskip 1em plus 0.5em minus
  0.4em\relax IEEE, 2011, pp. 1896--1903.

\bibitem{camacho1999model}
E.~F. Camacho, C.~Bordons, and M.~Johnson, ``Model predictive control. advanced
  textbooks in control and signal processing,'' 1999.

\bibitem{ljung1998system}
L.~Ljung, ``System identification,'' in \emph{Signal analysis and
  prediction}.\hskip 1em plus 0.5em minus 0.4em\relax Springer, 1998, pp.
  163--173.

\end{thebibliography}

\end{document}